\documentclass[journal]{IEEEtran}
\usepackage{amsmath}
\usepackage{amssymb}
\usepackage{amsthm}
\usepackage{xparse}
\usepackage{xspace}
\usepackage[T1]{fontenc} 

\NewDocumentCommand{\sym}{smO{}mo}{
	\expandafter\NewDocumentCommand\csname #2\endcsname{#3}{\ensuremath{{#4}}\xspace}
}

\NewDocumentCommand{\name}{mmm}{
	\expandafter\NewDocumentCommand\csname #1n\endcsname{}{#2\xspace}
	\expandafter\NewDocumentCommand\csname #1a\endcsname{}{#3\xspace}	
	\expandafter\NewDocumentCommand\csname #1ns\endcsname{}{#2s\xspace}
	\expandafter\NewDocumentCommand\csname #1as\endcsname{}{#3s\xspace}	
}

\name{lft}{Left-unique}{}
\name{fdemand}{Fixed-demand}{}
\name{acnet}{AC-network}{}
\name{sol}{solution}{solution}
\name{feas}{feasibility}{FEAS}
\name{opf}{optimal power flow}{OPF}
\name{mpf}{maximum potential flow}{MPF}
\name{ac}{AC}{AC}
\name{nph}{NP-hard}{}

\sym{powerset}[mO{}]{\mathcal{P}_{#2}(#1)}
\sym{net}[O{}O{}]{\mathcal{N}_{#1}^{#2}}[\acnetn]
\sym{feas}[O{\net}O{}]{{FEAS}_{#2}(#1)}
\sym{buses}[O{}O{}]{N_{#1}^{#2}}[set of buses]
\sym{lines}[O{}O{}]{E_{#1}^{#2}}[set of lines]
\sym{pflow}[O{}]{p_{#1}}[real line power flow]
\sym{qflow}[O{}]{q_{#1}}[reactive line power flow]
\sym{gensy}{G}
\sym{loadsy}{L}
\sym{pasy}{\Theta}
\sym{pa}[O{}]{\pasy_{#1}}[phase angle(s)]
\sym{pflowu}{P}
\sym{qflowu}{Q}
\sym{pdem}[O{}]{{\pflowu}_{#1}}[real power demand]
\sym{qdem}[O{}]{{\qflowu}_{#1}}[reactive power demand]
\sym{gens}[O{}]{\buses_{\gensy}^{#1}}[set of generators]
\sym{loads}[O{}]{\buses_{\loadsy}^{#1}}[set of loads]
\sym{deltamax}{\overline{\Delta}}[maximum phase angle difference]
\sym{ba}[O{}]{i_{#1}}[bus]
\sym{bb}[O{}]{j_{#1}}[bus]
\sym{ca}[O{}]{s_{#1}}[capacity]
\sym{co}[O{}]{g_{#1}}[conductance]
\sym{su}[O{}]{b_{#1}}[susceptance]
\sym{sedge}[O{\ba\bb}]{#1}
\sym{edge}[O{\ba\bb}O{\su}O{\co}]{\sedge[#1]_{#2}^{#3}}[line from \ba to (between) \bb with susceptance \su and conductance \co]
\sym{dlines}[O{\lines}]{{#1}^d}[set of lines, directed version]
\sym{npmax}[O{\co}O{\su}]{{\hat{\pflow}}}[real line power flow for phase angle difference of -\deltamax]
\sym{nqmax}[O{\co}O{\su}]{{\hat{\qflow}}}[reactive line power flow for phase angle difference of -\deltamax]

\newtheorem{lemma}{Lemma}
\newtheorem{theorem}{Theorem}

\title{AC-Feasibility on Tree Networks is NP-Hard}
\author{Karsten~Lehmann, Alban~Grastien, and Pascal~Van~Hentenryck \\
National ICT Australia (NICTA) and Australian National University, Canberra, Australia \\
Email: {first.last@nicta.com.au}
}

\begin{document}
\maketitle         

\begin{abstract}
  Recent years have witnessed significant interest in convex
  relaxations of the power flows, several papers showing that the
  second-order cone relaxation is tight for tree networks under
  various conditions on loads or voltages.  This paper shows that
  \aca-\feasn, i.e., to find whether some generator dispatch can
  satisfy a given demand, is NP-Hard for tree networks.
\end{abstract}

\section*{Nomenclature}
\begin{tabular}{ll}
    \net & \acnetn\\
    \buses & set of buses\\
    \gens & set of generators\\
    \loads & set of loads\\
    \ba& bus\\
    \bb& bus\\
    \lines & set of lines\\
    \dlines & set of lines with direction\\
    \edge & line from \ba to \bb\\ 
    \su& susceptance\\
    \co& conductance\\
    \ca& capacity\\
    \deltamax& maximum phase angle difference\\
    \pa & phase angle(s)\\
    \npmax & real line power flow for phase angle difference of\\
    &-\deltamax\\
    \nqmax & reactive line power flow for phase angle difference\\
    &off -\deltamax\\
    \pflow & real line power flow\\
    \qflow & reactive line power flow\\
    \pdem & real power demand\\
    \qdem & reactive power demand
\end{tabular}

\section{Introduction}

Many interesting applications in power systems, including optimal
power flows, optimize an objective function over the steady-state
power flow equations, which are nonlinear and nonconvex. These
applications typically include an \emph{\aca-\feasn} (\aca-\feasa)
subproblem: find whether some generator dispatch can satisfy a given
demand.

The first NP-hardness proof for \aca-\feasn was given for a cyclic
network structure in \cite{verma2009power}. It relies on a variant of the DC model
\cite{Stott:2009bb} but uses a sin function around the phase angle
difference. From an AC perspective, this means that conductances are
$0$, voltage magnitudes are all fixed at $1$, and reactive power is
ignored.  In recent years, there has been significant interest in
convex relaxations of the AC power flow equations following the
seminal work of Jabr, Lavaei, and Low \cite{Jabr06,lavaei2012zero}. Several
papers have shown that the second-order cone relaxation on {\em tree
  networks} is tight if load over-satisfaction is allowed
\cite{zhang2013geometry,li2012exact,lavaei2012zero}.  The second-order cone
relaxation is also tight on tree networks if the voltage bounds are
relaxed \cite{farivar2011inverter}.  Tree networks are important
obviously since they are the backbones of distribution systems.

This paper proves that \aca-\feasn is NP-Hard for tree networks. The
proof does not require bounds on generation and is valid for realistic
conductances, susceptances, and bounds on the phase angles. 

\section{Problem Definition}

This section presents the problem description and the assumptions
underlying the proof.  Our \aca-\feasn problem receives as input fixed
demands for real (\pdem) and reactive (\qdem) power. It fixes all
voltage magnitudes to one and assumes that lines have a maximum phase
angle difference $0< \deltamax \leq \pi/2$.  The proof also assumes a
susceptance $\su \leq 0$ and conductance $\co \geq 0$ and imposes a
natural condition on the relationship between $\su$, $\co$, and
$\deltamax$.

In the model, the set of buses \buses is defined as the disjoint union of 
the set of loads \loads and the set of generators \gens.
Hence every bus is either a generator or a load (with possibly $0$ demand).
$\lines \subseteq \powerset{\buses}[2]$ is
the set of lines and \dlines is the set of directed lines.

With these assumptions and notations, the \aca-\feasn problem consists
in finding the phase angles $\pa_i$, the real power flows
$\pflow_{ij}$, and the reactive power flows $\qflow_{ij}$ satisfying
\begin{align*}
    \forall \ba \in &\loads:\\
    &\sum_{\sedge \in \dlines} \pflow[\sedge] = \pdem[\ba]\\
    &\sum_{\sedge \in \dlines} \qflow[\sedge] = \qdem[\ba]\\
    \forall \ba \in &\gens:\\
    &\sum_{\sedge \in \dlines} \pflow[\sedge] \geq 0\\
    \forall \edge \in &\dlines:\\
    &\pflow[\sedge] = \co(1 - \cos(\pa[\ba]-\pa[\bb])) - \su\sin(\pa[\ba] - \pa[\bb])\\
    &\qflow[\sedge] = -\su(1 - \cos(\pa[\ba]-\pa[\bb])) - \co\sin(\pa[\ba] - \pa[\bb])\\
    &|\pa[\ba] - \pa[\bb]| \leq \deltamax.
\end{align*}

\noindent
This formulation uses phase angles and a bound on phase angles since
this makes the proof simpler. Phase angles are not typically used in
optimization over tree networks. However that there is no loss of
generality in this formulation, since imposing a maximum phase angle
difference is equivalent to enforcing a line capacity (thermal
limit). Indeed, the maximum phase angle difference \deltamax implies a
capacity of
\begin{align*}
\ca :=&
2(\co^2 + \su^2)(1 - \cos(\deltamax))\\
=&
(\co(1-\cos(\deltamax))-\su\sin(\deltamax))^2\\
&+ (\su(1-\cos(\deltamax))-\co\sin(\deltamax))^2.
\end{align*}

For a given capacity \ca and using that the phase angle difference has to be within $[-\pi/2, \pi/2]$ we can define a maximum phase angle difference \deltamax
$$
\deltamax :=
\begin{cases}
    \pi/2 & \text{if } s > 2(\su^2+\co^2)\\
    \arccos(1- \frac{s}{2(\su^2+\co^2)}) & \text{otherwise.}
\end{cases}
$$

\section{\aca-Feasibility on Star Networks is \nphn}

This section proves that the AC-feasibility of an \aca network with a
star structure and one load is NP-hard. The inspiration underlying the
proof came from the 2-bus example in \cite{bukhsh2013} that
exhibits disconnected feasibility regions.

Let $0 < \deltamax \leq \pi/2$.  The key element of the proof is that,
for any choice of \su and \co, the ratio between real and reactive
power is unique with respect to the phase angle difference. This is
captured in the following lemma, which also uses the following
notations for clarity:
\begin{align*}
\npmax &:= \co(1-\cos(-\deltamax))-\su\sin(-\deltamax)\\
\nqmax &:= -\su(1-\cos(-\deltamax))-\co\sin(-\deltamax).
\end{align*}

\begin{lemma}
\label{lemma}
Let \edge be a line with $\{\su, \co\} \neq \{0\}$ and $\deltamax \geq
\pa[\ba] - \pa[\bb] \geq 0$.  The following statements are true:
\begin{align}
    \pflow[\bb\ba]\nqmax&\leq \qflow[\bb\ba]\npmax;
    \label{eq:ratio2}\\
    \pflow[\bb\ba]\nqmax= \qflow[\bb\ba]\npmax &\iff \pa[\ba]-\pa[\bb] \in \{0, \deltamax\}.
    \label{eq:equal2}
\end{align}
\end{lemma}
\begin{proof}
To simplify notations we define $\Delta := \pa[\ba]-\pa[\bb]$; $t := \tan(-\Delta/2);$ $u := \tan(-\deltamax/2)$.
Let us assume that $\deltamax > \Delta > 0$.
Using the fact that the tangent is strongly monotonic increasing within the interval $(-\pi/4, 0)$ we have
\begin{align*}
u 
&<
t\\
u(\su^2+\co^2)
&<
t(\su^2+\co^2)\\
u\su^2 - t\co^2
&<
t\su^2 - u\co^2\\
u\su^2 - t\co^2 +\su\co(1-ut)
&<
t\su^2 - u\co^2 +\su\co(1-ut)\\
(\su-t\co)(u\su+\co)
&<
(\su-u\co)(t\su+\co)\\
(t\co-\su)(-u\su - \co)
&<
(u\co-\su)(-t\su - \co)
\end{align*}
Using the trigonometric identity 
$\tan(\alpha/2) = \frac{1-\cos(\alpha)}{\sin(\alpha)}$
and multiplying both sides of the last equation with 
$\sin(-\deltamax)\sin(-\Delta)$
(using the fact that $\Delta > 0$)
we get
\begin{align*}
(\co(1-\cos(-\Delta)) -\su\sin(-\Delta))\\
\cdot (-\su(1-\cos(-\deltamax)) -\co\sin(-\deltamax))\\
< 
(\co(1-\cos(-\deltamax))-\su\sin(-\deltamax) )\\
\cdot(-\su(1-\cos(-\Delta)) -\co\sin(-\Delta))
\end{align*}
which is 
$\pflow[\bb\ba]\nqmax < \qflow[\bb\ba] \npmax$ for $\deltamax > \Delta > 0$.
Eq.~\eqref{eq:ratio2} is true if $\Delta = 0$ or $\Delta = \deltamax$.
Hence Eq.~\eqref{eq:ratio2} and Eq.~\eqref{eq:equal2} are true in general.
\end{proof}

\noindent
To make sure that the load used in our encoding is in fact consuming
power, it is necessary to ensure that $\npmax < 0$.  This introduces a
constraint on the values of \deltamax, \su, and \co in the networks
considered by the proof. Note however that this constraint does not
remove realistic values for $b$, $g$, and $\deltamax$. The next lemma
establishes an important property of the phase angles derived from the
real power flow equation.

\begin{lemma}
    \label{lem:negative}
    Consider $0 < \deltamax$, $|\Delta| \leq \deltamax$, and \su
    and \co be such that the condition $\npmax < 0$ holds.  Then we have
    $$\co(1-\cos(\Delta)) - \su\sin(\Delta) \geq 0 \implies \Delta \geq 0.$$
\end{lemma}
\begin{proof}
  For $\Delta = 0$, we have $\co(1-\cos(-\Delta)) - \su\sin(-\Delta) =
  0$.  Assume that $\Delta < 0$.  We have
\begin{align*}
    0 &> \npmax = \co(1-\cos(-\deltamax)) - \su\sin(-\deltamax)\\
    0 &> \co(1-\cos(\deltamax)) + \su\sin(\deltamax)\\
    -\su\sin(\deltamax) &> \co(1-\cos(\deltamax))\\
    -\su &> \co\tan(\deltamax/2) \geq \co\tan(-\Delta/2)\\
    -\su &> \co\tan(-\Delta/2)\\
    -\su\sin(-\Delta) &> \co(1-\cos(-\Delta))\\
    0 &> \co(1-\cos(-\Delta)) + \su\sin(-\Delta)\\
    0 &> \co(1-\cos(\Delta)) - \su\sin(\Delta).
\end{align*}
This contradicts the premise that 
$\co(1-\cos(\Delta)) - \su\sin(\Delta) \geq 0$.
Hence we have $\Delta > 0$.
\end{proof}

\noindent
We are now in position to prove our main result.
\begin{theorem}
\aca-\feasn on trees is NP-hard.
\end{theorem}
\begin{proof}
  To prove that star networks are \nphn, we present a reduction from
  the \nphn \emph{subset sum} problem to \aca-\feasn.  Given a set $M
  \subset \mathbb{N}_{>0}$ and a number $w \in \mathbb{N}_{>0}$,
  the subset sum problem decides whether there exists $V
  \subseteq M$ such that $\sum_{x \in V} x = w$.  If such a set $V$
  exists, we call the problem instance $(M,w)$ solvable.

  Let $(M,w)$ be an arbitrary instance of the subset sum problem.  We
  define the \acnetn \net[M,w] via $\gens := M$; $\loads := \{l\}$;
  $\lines := \{\edge[xl][bx][gx] \mid x
  \in M \})$; $\pdem[l] := w\npmax$; $\qdem[l] := w\nqmax$ where
  $\deltamax$, $b$, and $g$ are chosen to satisfy the condition in
  Lemma~\ref{lem:negative}.\footnote{Observe that the susceptance and the
    conductance are given by $b x$ and $g x$ respectively for
    simplifying the proof.} This encoding is polynomial in the size of
  $(M,w)$, since it uses only rational numbers and finitely many real
  numbers constructed from rational numbers, sine, and cosine. The
  rest of the proof shows that
\begin{equation*}
\net[M,w] \text{ has feasible solution} \iff (M,w) \text{ is solvable}.
\end{equation*}

\noindent
Case 1: $\net[M,w] \text{ has feasible solution} \Longleftarrow (M,w) \text{ is solvable}$.\\
Let $V$ be a solution for $(M,w)$.  
We define $\pa[l] := 0$; 
$\forall
x \in V: \pa[x] := \deltamax, \pflow[lx] := x\npmax, \qflow[lx] :=
x\nqmax, 
\pflow[xl] := x\co(1-\cos(\deltamax))-x\su\sin(\deltamax), 
\qflow[xl] := -x\su(1-\cos(\deltamax))-x\co\sin(\deltamax)$; 
$\forall x \in
M\setminus V: \pa[x] := \pflow[lx] := \qflow[lx] := \pflow[xl] :=
\qflow[xl] := 0$.  It is easy to see that the maximum phase angle
difference constraints and the \aca-power laws are satisfied.  Using
the fact that $V$ is a solution for $(M,w)$, the conservation law at
$l$ is
\begin{align*}
\sum_{x \in M} \pflow[lx] 
= \sum_{x \in V} \pflow[lx]
= \sum_{x \in V} x\npmax
= w\npmax = \pdem[l]\\
\sum_{x \in M} \qflow[lx] 
= \sum_{x \in V} \qflow[lx]
= \sum_{x \in V} x\nqmax
= w\nqmax = \qdem[l].
\end{align*}
Moreover, the generation constraints are satisfied because $\co(1-\cos(\deltamax))-\su\sin(\deltamax)$ is always positive for a positive phase angle difference.
Hence we have defined a feasible solution. \\

\noindent
Case 2: $\net[M,w] \text{ has feasible solution} \implies (M,w) \text{ is solvable}$.\\
Let \pa, \pflow, and \qflow be the feasible solution.
Lemma~\ref{lem:negative} together with the fact that we have the constraint
that the real power at the generators has to be positive implies that
$\forall x \in M: \pa[x] - \pa[l] \geq 0$.  We define $V := \{ x \in M
\mid \pa[x]-\pa[l] > 0 \}$.  Because we have a feasible solution,
Kirchhoff's conservation law for real and reactive power becomes
$\sum_{x \in M} \pflow[lx] = w\npmax$ and $\sum_{x \in M} \qflow[lx] =
w\nqmax$. 
Using $\npmax < 0$ and 
$\deltamax > 0 \implies \nqmax > 0$
we can derive
\begin{align*}
    \sum_{x \in M} \frac{\pflow[lx]}{\npmax}
    &= \sum_{x \in M} \frac{\qflow[lx]}{\nqmax}
\end{align*}
\begin{align*}
    0 
    = \sum_{x \in M} (\frac{\pflow[lx]}{\npmax} - \frac{\qflow[lx]}{\nqmax})
    = \sum_{x \in V} (\frac{\pflow[lx]}{\npmax} - \frac{\qflow[lx]}{\nqmax})
    = \sum_{x \in V} (\pflow[lx]\nqmax - \qflow[lx]\npmax).
\end{align*}
Eq.~\eqref{eq:ratio2} in Lemma~\ref{lemma} implies that every summand in this
equation is non-positive.  Hence all summands must be 0.  Given our
choice of $V$ and using Eq.~\eqref{eq:equal2} from Lemma~\ref{lemma}, we have
$\forall x \in V: \pa[x]-\pa[l] = \deltamax$.  This implies $\forall x
\in V: \pflow[lx] = x\npmax$ and hence using Kirchhoff's conservation
law for real power we have $\sum_{x \in V} \pflow[lx] = \sum_{x \in V}
x\npmax = w\npmax$ which proves $\sum_{x \in V} x = w$.
\end{proof}

\section{Conclusion}

This paper has shown that AC-Feasibility on tree networks is NP-Hard,
indicating that convex relaxations cannot be tight on tree networks
without additional conditions on the network. The proof relies on the
existence of arbitrarily small bounds on voltage magnitudes (we fixed
the voltage magnitudes to 1 in the proof for simplicity) and either
generation bounds, capacity constraints, or a bound on phase angle
differences.

\bibliographystyle{IEEEtran}
\bibliography{ac_stars}

\begin{thebibliography}{1}
\providecommand{\url}[1]{#1}
\csname url@samestyle\endcsname
\providecommand{\newblock}{\relax}
\providecommand{\bibinfo}[2]{#2}
\providecommand{\BIBentrySTDinterwordspacing}{\spaceskip=0pt\relax}
\providecommand{\BIBentryALTinterwordstretchfactor}{4}
\providecommand{\BIBentryALTinterwordspacing}{\spaceskip=\fontdimen2\font plus
\BIBentryALTinterwordstretchfactor\fontdimen3\font minus
  \fontdimen4\font\relax}
\providecommand{\BIBforeignlanguage}[2]{{%
\expandafter\ifx\csname l@#1\endcsname\relax
\typeout{** WARNING: IEEEtran.bst: No hyphenation pattern has been}%
\typeout{** loaded for the language `#1'. Using the pattern for}%
\typeout{** the default language instead.}%
\else
\language=\csname l@#1\endcsname
\fi
#2}}
\providecommand{\BIBdecl}{\relax}
\BIBdecl

\bibitem{verma2009power}
A.~Verma, ``Power grid security analysis: An optimization approach,'' Ph.D.
  dissertation, Columbia University, 2009.

\bibitem{Stott:2009bb}
B.~Stott, J.~Jardim, and O.~Alsac, ``Dc power flow revisited,'' \emph{Power
  Systems, IEEE Transactions on}, vol.~24, no.~3, pp. 1290--1300, Aug 2009.

\bibitem{Jabr06}
R.~Jabr, ``Radial distribution load flow using conic programming,'' \emph{IEEE
  Transactions on Power Systems}, vol.~21, no.~3, pp. 1458--1459, Aug 2006.

\bibitem{lavaei2012zero}
J.~Lavaei and S.~H. Low, ``Zero duality gap in optimal power flow problem,''
  \emph{Power Systems, IEEE Transactions on}, vol.~27, no.~1, pp. 92--107,
  2012.

\bibitem{zhang2013geometry}
B.~Zhang and D.~Tse, ``Geometry of injection regions of power networks,''
  \emph{Power Systems, IEEE Transactions on}, vol.~28, no.~2, pp. 788--797,
  2013.

\bibitem{li2012exact}
N.~Li, L.~Chen, and S.~Low, ``Exact convex relaxation of {OPF} for radial
  networks using branch flow model,'' in \emph{Smart Grid Communications
  (SmartGridComm), 2012 IEEE International Conference on}, Nov 2012, pp. 7--12.

\bibitem{farivar2011inverter}
M.~Farivar, C.~R. Clarke, S.~H. Low, and K.~M. Chandy, ``Inverter var control
  for distribution systems with renewables,'' in \emph{Smart Grid
  Communications (SmartGridComm), 2011 IEEE International Conference on}.\hskip
  1em plus 0.5em minus 0.4em\relax IEEE, 2011, pp. 457--462.

\bibitem{bukhsh2013}
W.~Bukhsh, A.~Grothey, K.~McKinnon, and P.~Trodden, ``Local solutions of the
  optimal power flow problem,'' \emph{Power Systems, IEEE Transactions on},
  vol.~28, no.~4, pp. 4780--4788, Nov 2013.

\end{thebibliography}

\end{document}